\documentclass[journal]{IEEEtran}

\ifCLASSINFOpdf
 \else
 \fi

\usepackage{algorithm}
\usepackage{mdwmath}
\usepackage{mdwtab}
\usepackage{url}
\usepackage{algpseudocode}
\usepackage{algorithm}

\usepackage{latexsym,amsmath,amsfonts,amsthm}
    \usepackage{amssymb}
\usepackage[colorlinks=true]{hyperref}
\usepackage{color,cite}
\usepackage[dvipsnames]{xcolor}
\usepackage{mathtools}
\newtheorem{theorem}{Theorem}

\newtheorem{proposition}[theorem]{Proposition}
\newtheorem{lemma}[theorem]{Lemma}

\theoremstyle{definition}
\newtheorem{example}[theorem]{Example}
\newtheorem{definition}[theorem]{Definition}

\theoremstyle{remark}
\newtheorem{remark}[theorem]{Remark}

\newcommand{\R}{\mathbb{R}}
\newcommand{\N}{\mathbb{N}}
\newcommand{\Z}{\mathbb{Z}}

\newcommand{\la}{\langle}
\newcommand{\ra}{\rangle}

\newcommand{\be}{\begin{equation}}
\newcommand{\ee}{\end{equation}}

\begin{document}
%
\title{On the finiteness of accessibility test for nonlinear discrete-time systems}
%
%
%

\author{{Mohammad Amin Sarafrazi, 
       Ewa Pawluszewicz, Zbigniew Bartosiewicz 
        and \"{U}lle Kotta}
\thanks{Mohammad Amin Sarafrazi: 
 Rezvan complex,
Motahari Sq., Motahari Blvd., 71868-98544  Shiraz, Iran}
\thanks{E.Pawluszewicz is with Faculty of Mechanical Engineering, Bialystok University of Technology, 15-351 Bialystok, Poland}
\thanks{Z.Bartosiewicz is with Faculty of Computer Science, Bialystok University of Technology, 15-351 Bialystok, Poland}
\thanks{U.Kotta is with Department of Software Science, Tallinn University of Technology, 12616 Tallinn, Estonia}} 


\maketitle

\begin{abstract}

It is shown that for two  large subclasses of discrete-time nonlinear systems - analytic systems defined on a compact state space and  rational systems - the minimum length  $r^*$ for input sequences, called here accessibility index of the system, can be found, such that from any point $x$, system is accessible iff it is accessible for input sequences of length $r^*$. Algorithms are presented to compute $r^*$, as well as an upper bound for it, which can be computed easier, and hence provide  finite tests for determination of accessibility. The algorithms  also show how to construct the set of points from which the  system is not accessible in any finite number of steps. Finally, some  relations between  generic accessibility  of the system and accessibility of individual points in finite steps are given.
\end{abstract}



\IEEEpeerreviewmaketitle

\section{Introduction}

Different variants of accessibility property, being closely related  to controllability, are useful concepts in nonlinear control theory,  whereas in the linear case all of them coincide. 
Roughly speaking,   forward  accessibility from a point $x_0$ concerns the ability of input sequences to move the end-point state in  all directions of some open subset of state space, and backward accessibility from a point $x_0$ concerns the ability of input sequences to move the state from all points of some open subset of the state space toward $x_0$. Accessibility is determined by the rank of the accessibility distribution (see \cite{Jakubczyk90} for discrete-time  systems) computed at the initial state.
 In almost all control methods, it is assumed that the initial state is a regular point of accessibility distribution, meaning that the accessibility distribution has constant rank in a neighborhood of the initial state.  But, checking the validity of this assumption is not a trivial task.
 Specific for nonlinear systems, it  is necessary to distinguish between pointwise accessibility (or shortly accessibility), and generic accessibility (i.e. accessibility from \emph{almost} any point). Various finite criteria exist  for generic accessibility  \cite{Aranda}, \cite{Mullari_TAC}, that require no more  than $n$ steps, where $n$ is the state  dimension. However, for pointwise accessibility, in general, it is  not clear a priori how many steps are needed to distinguish between  the accessibility singular points (from which the system is not  accessible in any finite number of steps) and the points from which  the system is forward accessible in a finite number of steps.   Moreover, it is even unknown whether an upper bound exists for the  number of steps needed. The similar problem exists in the  continuous-time case too \cite{Kawski}.   
While the weaker property of  accessibility in $k$ steps can be checked by a finite test \cite{Mullari_sub}, \cite{Rieger}, the nonlinear system may become accessible from a given point $x$, only when the length of input sequence is larger than some unknown integer $r^*$ that can be far larger than the state space dimension. This is why the accessibility rank condition is usually expressed in terms of an infinite set of  conditions. \cite{Jakubczyk90}, \cite{Mullari_TAC}. 

 \color{black}

\begin{example}\label{ex_Introduction}

Consider the  nonlinear discrete-time system described by 
$ x_1^{\langle 1\rangle}(t)  =  x_2(t)$
, and
 $x_2^{\langle 1\rangle}(t) =  -x_1(t)+x_2(t)+u(t)x_2^2(t)-ux_2(t)$, where the notation ${\langle 1\rangle}$ denotes the value of the variable in the next time instance. By the results of \cite{Aranda} this system is generically   accessible, meaning that starting from \emph{almost} any  initial state,   the set  of possible end point states  in 2 steps  for all input pairs $(u(0),u(1))$ contains non-empty open subsets  of $\mathbb{R}^2$. However, starting from the initial state  $x_0=(0,1)^T$ one can  verify that, in spite of the chosen input pairs,  the trajectory is  $x(1)=(1,1)^T, x(2)=(1,0)^T$. Hence the system is not forward accessible  in one or two steps from  $x_0$. Further computation of  $x(3)=(0,-1)^T$, $x(4)=(-1,-1+2u(3))^T$, shows that in 4 steps, the input sequence affects the end point state only in one direction, hence the system is not forward accessible up to 4 steps. Computing further, we obtain $x(5)=(-1+2u(3),2u(3)+u(4)(2u(3)-1)(2u(3)-2))^T$, which shows  that $x(5)$ can be moved in two independent directions.  To conclude, the system  is not forward accessible from $x_0$ in four steps or less, but eventually it becomes forward accessible in five steps.
\end{example}

The goal of this paper is to   provide a finite test for deciding accessibility (in any number of steps), by finding a  certain integer $r^*$, called accessibility index of the system, such that for all points of state space, examination of accessibility with respect to input sequences of length $r^*$  determines accessibility for any input sequence. To this aim, instead of assessing every individual point of state space separately, we look at the "big picture" of the entire set of accessibility singular points, by assigning certain algebraic sets $S_k$ to the points that are not accessible for inputs of length $k$. Using the language of ideals and varieties, and some geometric properties of the system relying on analyticity and submersivity, we find  the smallest $k$ such that $S_k$ is an invariant set. This $k$ is the accessibility index of the system. As a byproduct, our method also gives the   entire set of accessibility singular points,  using a single chain of strictly ascending chain of ideals.

The ideas, based on ideals and algebraic geometry were
already successfully used in \cite{Bartosiewicz_1995}, \cite{Bartosiewicz_2016}, \cite{Kawano_2013} for characterizing
different obervability properties, in \cite{Kawano_2013_1},  \cite{Nesic_1999}, \cite{Nesic_1998} for
controllability and accessibility of   polynomial systems. In \cite{Nesic_1_1998} a sufficient condition was given for a polynomial system to be  accessible from all points under some assumptions. In  \cite{Kawano_2016}  similar problem was studied and a finite criterion for accessibility from an equilibrium point was introduced.
However, there appears to be a gap in their results on  accessibility (see Remark \ref{remark2}). \\ To the best of our knowledge, our  results provide the first finite test for accessibility in any number of steps, for rational systems and analytic systems on compact semianalytic sets.


The paper is organized as follows. Preliminaries and definitions are given in Section \ref{II}. 
Accessibility criterion, and relation between generic accessibility and pointwise accessibility are obtained in Section \ref{III}.
The main results of the paper are  given in Section  \ref{IV} that investigates separately the cases of  rational systems and analytic systems with compact state space. Appendix presents algorithms that are used in examples.
\color{black}

\section{Preliminaries}\label{II}

Consider the discrete-time nonlinear system
\begin{equation}\label{system}
 x^{\langle 1\rangle}(t)=\Phi(x(t),u(t)),\;\;\;x(t_0)=x_0
\end{equation}
where $x^{\langle 1\rangle}(t):=x(t+1)$ denotes the first order forward shift of the state vector $x(t)$ for any $t\in\Z$, $x(t)\in X\subseteq\R^n$, $u(t)\in U\subseteq\R^m$, and $X,U$  are open subsets, with the standard topology on $\R^n$. Throughout the paper, we assume that the state transition map $\Phi:X \times U\rightarrow X$ is an analytic function. 
 For a matrix with analytic entries, we will use the terms \emph{``generic rank''} and \emph{``maximal rank''}  interchangeably, since for such matrices these two are equal.
The system is assumed to be generically submersive \footnote{The assumption (\ref{submersivity}) is not restrictive, it is a necessary condition for accessibility \cite{Grizzle93}},  which means that
\begin{equation}\label{submersivity}
 {\rm rank}\left[\frac{\partial\Phi(x,u)}{\partial(x,u)}\right]=n
\end{equation}
holds almost everywhere (i.e. except on a zero-measure set) on $X \times U$.

Let $\cal K$ be the field of analytic functions in a finite number of independent variables from the infinite set
$\mathcal{C}=\{x,u^{\langle k\rangle} ,k\geq 0\}$. For a function $\varphi \in \cal K$, the forward shift of $\varphi(x,u,\ldots,u^{\langle k\rangle})$ is defined as
\begin{equation}\label{shift}
 \varphi^{\langle 1\rangle}(x,u,u^{\langle 1\rangle},\ldots,u^{\langle k+1\rangle}):=\varphi(\Phi(x,u),u^{\langle 1\rangle},\ldots,u^{\langle k+1\rangle}).
\end{equation}

Let $U^k=\underbrace{U\times U\times \ldots \times U}_k$ denote the space of control sequences  $\textbf{u}=\{u(0),\ldots,u(k-1)\}$.
For 
initial point $x_0\in X$ and $k\geq1$ we define the map $\Phi_{x_0}^k:U^k\rightarrow X$ as\[  \label{Phi^k}
 \Phi_{x_0}^k (\textbf{u}):=\Phi(\Phi(\ldots(\Phi(x_0,u_0),u_1 )\ldots),u_{k-1}).
\]
The image of the map $\Phi_{x_0}^k$, denoted by $\mathcal{A}_k^+(x_0)$ and called the \emph{forward-attainable set in $k$ steps from $x_0$}, describes the set of states that can be obtained by evaluating the solution of the system (\ref{system}) in $k$ steps from the initial state $x_0$ using all control sequence $\textbf{u}\in U^k$. 
 By
\begin{equation}\label{A+}
 \mathcal{A}^+(x_0)=\bigcup_{k\in\N\ }\mathcal{A}_k^+(x_0)
\end{equation}
we will denote the \emph{forward-attainable set in finite number of steps} from $x_0$ for the system (\ref{system}). Additionally, if $W\subseteq U^k$, then
$ \Phi_{x_0}^k(W):= \bigcup_{\textbf{u}\in W}\Phi_{x_0}^k(\textbf{u})$.

\begin{definition}
 The system (\ref{system}) is said to be \emph{(forward) accessible in $k$ steps from the state $x_0\in X$}, if the interior of the forward-attainable set $\mathcal{A}_k^+(x_0)$ is nonempty. Otherwise it is \emph{non-accessible in $k$ steps from $x_0$}. The system (\ref{system}) is said to be \emph{(forward) accessible from $x_0$} (in finite number of  steps), if ${\rm int}\mathcal{A}^+(x_0)\neq \emptyset$. This system is said to be non-accessible from $x_0$ if it is not accessible from $x_0$.
\end{definition}

\begin{definition}\label{generically forward accessibility}
The system (\ref{system}) is called \emph{accessible (pointwise accessible) from a set $R\subset X$} if it is  accessible from any point of this set.  System (\ref{system}) is called \emph{generically  accessible} if there is an open and dense subset $D$ of $X$ such that the system is  accessible from $D$.
\end{definition}

\begin{definition}
 The (finite) integer $r^*$ is called the \emph{accessibility index} of the system (\ref{system}) over $X$ if $r^*$ is the maximum integer
for which there exists at least one point $x_0 \in X$  such that the system remains  non-accessible  from $x_0$ up to $r^*-1$ steps, but becomes accessible in $r^*$ steps. If there is no such finite integer $r^*$, then $r^* = \infty$.
	\end{definition}
	
	\section{Accessibility Criteria}\label{III}

 Various accessibility criteria can be found in literature  (for example see  \cite{Jakubczyk90},  \cite{Albertini}, \cite{Rieger}), which are more or less similar to the following: \color{black}
\begin{proposition}\cite{Jakubczyk90}\label{Jakubczyk}
 For any fixed  $x_0 \in X$ the interior of the 
  forward-attainable set $\mathcal{A}_k^+(x_0)$ is nonempty if and only if
\[
  {\rm sup}\{{\rm rank}\frac{\partial}{\partial\mathbf{u}}\Phi_{x_0}^k(\mathbf{u}),\;\mathbf{u}\in U^k\}=n
\]
and thus  the  necessary and sufficient condition for accessibility of system (\ref{system}) from $x_0$ is \be \label{rankcondition}
 {\rm sup}\{{\rm rank}\frac{\partial}{\partial\mathbf{u}}\Phi_{x_0}^k(\mathbf{u}),\;\mathbf{u}\in U^k,k\geq 1\}=n.
\ee
\end{proposition}

The first step toward the simplification of checking the infinite set of     conditions (\ref{rankcondition}) is to show that, under the submersivity assumption (\ref{submersivity}), the  non-accessibility in (exactly) $k$ steps  is equivalent to  non-accessibility up to $k$ steps .

\begin{lemma}\label{lemma1}
 Assume that the system (\ref{system}) is generically submersive. If it is  accessible from the state $x_0$ in $k$ steps, then it is  accessible in $r$ steps too, for any integer $r>k$.
\end{lemma}
\begin{proof}
 Suppose that the system (\ref{system}) is  accessible from $x_0$ in $k$ steps. Thus, by definition, $\Phi_{x_0}^k(U^k)$ contains an open subset of $X$. First we show that the set $\Phi_{x_0}^{k+1}(U^{k+1})$ contains an open subset of $X$, too. The set $\Phi_{x_0}^{k+1}(U^{k+1})$ can be written as
 \begin{equation}\label{Phi_1}
 \Phi_{x_0}^{k+1}(U^{k+1})=\{\Phi(y,u):\;y\in \Phi_{x_0}^k(U^k)\;\text{and}\;u\in U\}.
 \end{equation}
Since, by assumption, $\Phi_{x_0}^k(U^k)$ and $U$ contain open subsets of $X$ and $U$, respectively, then from the generic submersivity assumption and from local surjectivity theorem \cite{Abraham88}, there exists $y_0\in {\rm int}(\Phi_{x_0}^k(U^k))$ and an input $u_0\in U$ such that the map $\Phi$ is locally surjective around the point $(y_0,u_0)$. This means that, there exists an open neighborhood of $y_0$ denoted as
$\beta(y_0)\subseteq{\rm int}(\Phi_{x_0}^k(U^k))$, such that $\{ \Phi_{y^*}(U),~ y^* \in \beta(y_0) \}$ contains an open subset of $X$. Hence from (\ref{Phi_1}) we conclude that the set $\Phi_{x_0}^{k+1}(U^{k+1})$ contains an open subset of $X$. This means that the system (\ref{system}) is  accessible from $x_0$ in $k+1$ steps. A simple induction establishes the assertion.
\end{proof}

The criteria (\ref{rankcondition}) can be stated in the form of rank of a series of matrices,  similar to the  controllability rank condition for the time-varying linear discrete-time systems \cite{Weiss}. \\
Let
$A(x,{\rm u}):=\frac{\partial\Phi(x,u)}{\partial x}$, $B(x,u):=\frac{\partial\Phi(x,u)}{\partial u}$
and put
\begin{eqnarray}\label{M_k}
 M_1(x,u)  & := & B(x,u), \nonumber\\
 M_k(x,\mathbf{u}) & := & [A^{\langle k-1\rangle}M_{k-1}|B^{\langle k-1\rangle}](x,\mathbf{u}),
\end{eqnarray}
$ k=2,3,4,\ldots$, where forward shifts of respective matrices are defined elementwise according to the rule given by (\ref{shift}). 
 
\begin{lemma}\label{lemma2}
 The system (\ref{system}) is not  accessible up to $k$ steps from a point $x_0\in X$  if and only if
\[
  {\rm rank}(M_k(x_0,\mathbf{u}))<n\;\text{for\;all}\;{\bf u}\in U^k.
\]
\end{lemma}
\begin{proof}
 From Proposition \ref{Jakubczyk} it follows that the system (\ref{system}) is  accessible from $x_0$ in $k$ steps if and only if the matrix $\frac{\partial}{\partial {u}}\Phi_{x_0}^k(\textbf{u})$ has generic rank $n$ over all $\textbf{u} \in U^k$. 
 Differentiating the map $\Phi_{x_0}^k(\textbf{u})$ and using the chain rule, we obtain
 $ \frac{\partial}{\partial \mathbf{u}}\Phi_{x_0}^k(\mathbf{u})
                                                              = [A^{\langle k-1\rangle}A^{\langle k-2\rangle}\ldots A^{\langle 1\rangle}B|
                                                                A^{\langle k-1\rangle}A^{\langle k-2\rangle}\ldots A^{\langle 2\rangle}B^{\langle 1\rangle}| \ldots\\
                                                          \ldots|
                                                                A^{\langle k-1\rangle}B^{\langle k-2\rangle}|B^{\langle k-1\rangle}](x_0,\mathbf{u})
                                                          =  M_k(x_0,\mathbf{u})$,
which establishes the claim.
\end{proof}

Before proceeding to the main result of the paper, we single  out  systems that are not generically accessible, by showing that  such systems are non-accessible anywhere.
\begin{theorem}\label{theorem1}
The system (\ref{system}) is generically  accessible if and only if  $M_{n}(x,\mathbf{u})$ has generic rank $n$ on $X\times U^{n}$. Moreover, if the system (\ref{system}) is not generically accessible, then it is  non-accessible from any $x$ and for any number of steps. 
\end{theorem} 
\begin{proof}
For simplicity, we prove the theorem for the single input case. If the generic rank of $M_{n}$ becomes $n$, then the system (\ref{system}) is accessible from almost all points, in $n$ steps, and more,  via Lemma \ref{lemma1}. 

 For the converse, we use  a deduction similar to the proof of  the Cayley-Hamilton theorem. Assume that $M_n$ has generic rank less than $n$. Then, considering (\ref{M_k}), there exists a vector $P=[p_1,…,p_n ]^T \in {\cal K }^n$ such that
 \be \label{KH} B^{\la n-1\ra} =A^{\la n-1\ra }  M_{n-1} P\ee 
which, after forward shifting, becomes
\be \label{KH1} B^{\la n \ra} =A^{\la n \ra}  M_{n-1}^{\la 1\ra}  P^{\la 1\ra}.\ee 
Multiplying both sides of (\ref{KH}) by $A^{\la n \ra}$  gives
\be  \label{KH2} A^{\la n \ra}  B^{\la n-1\ra} =A^{\la n \ra }  A^{\la n-1\ra } M_{n-1} P. \ee
Now, note that the left hand sides of (\ref{KH1}) and  (\ref{KH2})  are the last two columns of $M_{n+1}$, while the right hand sides of (\ref{KH1}) and  (\ref{KH2}) are linear combinations of the first $n-1$ columns of $M_{n+1}$, thus the generic rank of $M_{n+1}$ is less than  $n$. A simple induction proves that also all $M_k$ for $ k\geq n$ has  generic rank less than $n$. Because the generic rank of a  matrix with analytic entries is  its maximum rank,  from Proposition \ref{Jakubczyk} it follows that $\mathcal{A} _k^+ (x)$ has empty interior for all $k$ and all $x\in X$. In this case, from (\ref{A+})  using  Baire category theorem \cite{Abraham88} we conclude that $\mathcal{A}^+ (x)$ has empty interior for every $x\in X$.
\end{proof} \color{black}

Theorem \ref{theorem1} shows that generic accessibility of the system is necessary for  accessibility from any point $x_0$.  This property can be checked by a (finite) rank test on the matrix $M_{n}(x,\mathbf{u})$. Therefore, in the rest of the paper, we assume that the system is generically accessible.

Denote by $S_0:=X$ and by $S_k\subseteq X$ the sets of all states from which the system (\ref{system}) is not accessible up to $k$ steps. It means that for any $x_0\in S_k$, sets $\Phi_{x_0}^i(U^i)$ for $1\leq i\leq k$ do not contain an open subset of $X$. Additionally, by $S_\infty$ we denote the set of points from which system is not accessible in any finite number of steps. These points are called \emph{accessibility singular points}.
Thanks to Lemma \ref{lemma1}, we have the following descending chain of subsets:
\begin{equation}\label{descending_seq}
 S_n\supseteq S_{n+1}\supseteq S_{n+2}\supseteq \ldots \supseteq S_\infty.
\end{equation}
By construction, we obtain that $\Phi(S_{k+1},U)\subseteq S_k$.



Based on Theorem \ref{theorem1}, for an analytic system (\ref{system}) that is not generically accessible it holds that $S_\infty=X$, i.e., the system (\ref{system}) is   non-accessible from all points of $X$.  \\
\color{black}

\section{Finiteness conditions of accessibility}\label{IV}
Let us recall some basic facts from ideal theory. An \emph{ideal} $I$ of a commutative ring $\mathcal{P}$ is a subset of $\mathcal{P}$ with the properties that if $p\in\mathcal{P}$ and $a,b\in I$ then $a+b\in I$ and $ap\in I$. The  \emph{radical} of an ideal $I$ of $\mathcal{P}$, denoted by $\sqrt{I}$, is the set $\{p\in\mathcal{P}:\;p^n\in I\;\text{for\;some}\;n\in\N\}$. If $I$ coincides with its own radical, then $I$ is called a \emph{radical ideal}. The \emph{real radical} of $I$, denoted by $\sqrt[\R]{I}$, is the set of all $p\in\mathcal{P}$ for which there are natural $m,k$ and $q_1,\ldots,q_k\in\mathcal{P}$ such that $p^{2m}+\sum_i^kq_i^2\in I$. If $I,J$ are ideals of $\mathcal{P}$, then $(i)$ the real radical of $I$ is an ideal of $\mathcal{P}$, $(ii)$ $I\subseteq\sqrt{I}\subseteq\sqrt[\R]{I}$, $(iii)$ if $I\subset J$ then $\sqrt[\R]{I}\subseteq\sqrt[\R]{J}$, see \cite{Bochnak}.
A semialgebraic (respectively semianalytic) set $X$ is a set such that for every $x\in X$ there is an open neighborhood $V$ of $x$ with property that $X\cap V$ is a finite Boolean combination of sets $\{\overline{x}\in V:\varphi(\overline{x})=0\}$ and $\{\overline{x}\in V:\psi(\overline{x})>0\}$ where $\varphi,\psi:V\rightarrow\R$ are polynomial functions  (respectively  analytic functions). For a set $A \subset X$, the Zariski closure of $A$  is defined as the smallest algebraic variety containing A,  and is denoted by $\overline{A}$.

To study further forward accessibility of system (\ref{system}) we  need additional assumptions. We will address two cases:\\
\textit{A.} function $\Phi$ is rational defined on $X\times U$ \\
\textit{B.} function $\Phi$ is analytic defined on $X\times U$ with a compact semianalytic $X$.


\subsection{Rational case}
Let $\R[x]=\R[x_1,\ldots,x_n]$ be the commutative ring of polynomials in $x$. Recall that for an ideal $I$ of $\R[x]$ its \emph{zero-set} is defined as $\mathcal{Z}(I):=\{x\in\R^n:p(x)=0\;\text{for\;all}\;p\in I\},$ i.e. $\mathcal{Z}(I)$ is the set of common zeroes of all polynomials in $I$. If $\mathfrak{Z}\subset\R^n$, then the \emph{zero-ideal} of $\mathfrak{Z}$, denoted by $\mathcal{I}(\mathfrak{Z})$, is defined as the set
$\mathcal{I}(\mathfrak{Z})=\{p\in\R[x]:p(x)=0\;\text{for\;all}\;x\in \mathfrak{Z}\}.$
If $I$ is an ideal of $\R[x]$ and $\mathfrak{Z}\in\R^n$, then $(i)$ $\mathcal{I}(\mathcal{Z}(I))=\sqrt[\R]{I}$, $(ii)$ $\mathfrak{Z}\subseteq \mathcal{Z}(\mathcal{I}(\mathfrak{Z}))$, $(iii)$ $I\subseteq \mathcal{I}(\mathcal{Z}(I))$, $(iv)$ $\mathcal{Z}(\mathcal{I}(\mathcal{Z}(I)))=\mathcal{Z}(I)$ and $(v)$ $\mathcal{I}(\mathcal{Z}(\mathcal{I}(\mathfrak{Z})))=\mathcal{I}(\mathfrak{Z})$, see \cite{Bochnak}.

\begin{theorem}\label{theorem2}
 Assume that the system (\ref{system}) is rational and generically accessible. Then under submersivity assumption (\ref{submersivity})
\begin{enumerate}
\item [(i)] The descending chain of sets (\ref{descending_seq}) eventually stabilizes.
\item [(ii)] The inclusion relation in  the descending chain of sets (\ref{descending_seq}) before stabilization  is strict.
\item [(iii)] The system has  a finite accessibility index $r^*$, where  $r^*$ is the smallest integer such that $S_{r^*}=S_{r^* +1}$.
\end{enumerate}
\end{theorem}
\begin{proof}
 Let $\R_x [\mathbf{u}]$ be the commutative ring of polynomials in $\mathbf{u}=\{u_0,\ldots,u_{k-1}\}$ with coefficients in the ring $\R[x]$.
 For a fixed $k>n$, let $i_k$ be the number of all $n \times n$  submatrices of $M_k(x,\mathbf{u})$. Denote by $p_{k,1},\ldots,,p_{k,i_k}$ the numerator of determinants of these submatrices as polynomials in $R_x [\mathbf{u}]$. So, each $p_{k,i}$ can be written as
 \begin{equation}\label{p}
   p_{k,i}(x,\mathbf{u})=\sum_{j=1}^{j_{k,i}}a_{k,i,j}(x)b_{k,i,j}(\mathbf{u})
 \end{equation}
 where each $b_{k,i,j}(\mathbf{u})$ is a monomial in $\mathbf{u}$ with coefficient $a_{k,i,j}(x)$ in $ \mathbb{R}[x]$. The index $j_{k,i}$ in (\ref{p}) denotes the number of distinct monomials in $\mathbf{u}$.

 By analyticity and using Lemma \ref{lemma1}, a point $x_0\in X$ belongs to the set $S_k$ if and only if $a_{k,i,j}(x_0)=0$ for all $1\leq i\leq i_k$ and $1\leq j\leq j_{k,i}$. Thus every set $S_k$ is an algebraic set in  $X$. Since the dimension of $X$ is finite, the descending chain of algebraic sets (\ref{descending_seq}) stabilizes after finitely many steps \cite{Lang}, i.e. there exists an integer $r^*$ such that
 $ S_{n}\supseteq S_{n+1}\supseteq \ldots \supseteq S_{r^*}=S_{r^*+1}=\ldots=S_\infty.$  Hence the part (i) is proven.

 Assume that the integer $l\geq n$ is the smallest integer such that $S_l=S_{l+1}$. Since  $\Phi(S_{l+1},U)\subseteq S_l$ for any $\mathbf{u}\in U^k$, equality $S_{l+1}=S_l$ implies that $\Phi(S_l,U)\subseteq S_l$, i.e. $S_l$ is closed with respect to the multivalued map $\Phi(\cdot,U)$. Now, by assuming generic accessibility, $S_l$ is a zero measure set. Hence the image of a point $x\in S_l$ under any sequence of inputs from $U^k$, belongs to a set of measure zero. Thus all points $x\in S_l$ are among accessibility singular points, or say, $S_l\subset S_\infty$. Obviously, from Lemma \ref{lemma1} we have $S_\infty \subset S_l$. Thus $S_l=S_{r^* }=S_\infty$ and we have the following strict inclusion property
 $ S_{n}\supsetneq S_{n+1}\supsetneq \ldots \supsetneq S_{r^*}=S_{r^*+1}=\ldots=S_\infty$, which proves the part (ii) of the theorem. 
 Part (iii) is a trivial conclusion of (ii). 
\end{proof}
As a consequence of  Theorem \ref{theorem2} we have the ascending chain
\[
 I_{n}\subsetneq I_{n+1}\subsetneq\ldots\subsetneq I_{r^*}=I_{r^*+1}=\cdots
\]
  of ideals in $\R[x]$
 where $I_{k}= \mathcal{I}(S_{k})$. Let $I_{M_k}:=\langle a_{k,1,1},\ldots,a_{k,i_k,j_{k,{i_k}}}\rangle$  where $\langle a_{k,1,1},\ldots,a_{k,i_k,j_{k,{i_k}}}\rangle$ denotes the ideal generated by $a_{k,1,1},\ldots,a_{k,i_k,j_{k,{i_k}}}$, and $a_{k,i_k,j_{k,i_k}}$ are as in (\ref{p}). Property $\mathcal{I}(\mathcal{Z}(I))=\sqrt[\R]{I}$ and Lemma \ref{lemma1} assures that
$ I_k=\sqrt[\R]{I_{M_k}}$. This allows the usage of computer algebra tools for  obtaining the $S_\infty$ and respectively the accessibility index $r^*$, by checking equality of ideals. Namely, $r^*$ is the smallest integer such that $I_{r^*}=I_{r^*+1}$.  Algorithm 1 in Appendix is  based on Theorem  \ref{theorem2}.

 Computation of ideals $I_k$ involves computation of real radical ideals, which is computationally challenging. So we introduce an alternative approach. 
 
Consider another chain of ideals associated with the system (\ref{system}) that satisfies conditions of Theorem \ref{theorem2}. Namely, consider the ideals
 $\overline{I}_k :=  \bigcup_{l=n}^k I_{M_l}$
 for any $k\geq n$. By construction, we have the ascending chain
 \be \label{14}
  \overline{I}_{n}\subseteq \overline{I}_{n+1}\subseteq\ldots 
  \ee
   of ideals in $\R[x]$.


\begin{theorem}\label{theorem12}
  Assume that the system (\ref{system}) is  rational and generically accessible. Under submersivity assumption (\ref{submersivity}), there exists a finite integer $\kappa$ such that $\overline{I}_\kappa=\overline{I}_{\kappa+1}$. Moreover, $r^* \leq \kappa$ and $\mathcal{Z}(\overline{I}_\kappa)=S_{\infty}$.
\end{theorem}
\begin{proof}

Since $\R[x]$ is a Noetherian ring, from the Hilbert Basis Theorem \cite{Cox} it follows that there exists some $\kappa$ such that $\overline{I}_\kappa=\overline{I}_{\kappa+1}$. 
 Now, from the proof of Theorem \ref{theorem2} $S_k=\mathcal{Z}\langle a_{k,1,1},\ldots,a_{k,i_k,j_{k,{i_k}}}\rangle $. From the fact that
 $\mathcal{Z}(\bigcup_i I_i)=\bigcap_i\mathcal{Z}(I_i)$ (see \cite{Cox}),
 definition of sets $S_k$ and Theorem \ref{theorem2} it follows that
  $\mathcal{Z}(\overline{I}_k)=\bigcap_{l=k^*}^k S_l=S_k$.
Hence, $\mathcal{Z}(\overline{I}_k)=S_k=\mathcal{Z}(I_k)$.
 Then from  $\overline{I}_\kappa=\overline{I}_{\kappa+1}$ we obtain $S_\kappa=S_{\kappa+1}$, which using Theorem \ref{theorem2} gives $\mathcal{Z}(\overline{I}_\kappa)=S_{\infty}$. Since $r^*$ is the smallest integer such that $S_{r^*}=S_{r^*+1}$, therefore from $S_\kappa=S_{\kappa+1}$ we have $r^* \leq \kappa$.
\end{proof}

The (finite) integer $\kappa$ in Theorem \ref{theorem12} serves as an upper bound for the accessibility index $r^*$, while its computation does not involve the real radical ideals, and therefore it may be easier to compute.  Algorithm 2 in Appendix is based on Theorem \ref{theorem12}.

\begin{example}
 Let us consider a current-controlled coil, rotating in a homogeneous magnetic field, as described in \cite{Mullari_sub}. The  system equations after  Euler's discretization with step $T$           
   is
 \begin{equation}\label{system_after_dyscretization}
  x_1^{\langle 1\rangle}  =  x_1+Tx_2,\;\;\;
  x_2^{\langle 1\rangle}  =  x_2+T(ax_1u-bx_2)
 \end{equation}
where $x_1$ and $x_2$ denote the rotational angle and the angular velocity of the coil, respectively. It is easy to see that system (\ref{system_after_dyscretization}) is generically accessible (in 2 steps), and  $I_{M_2}=\langle x_1(x_1+Tx_2)\rangle$, shows that from the points of $\mathcal{V}(I_{M_2})$ the system (\ref{system_after_dyscretization}) is not accessible in 2 steps. 
 Using Algorithm 2, we obtain $\bar{I}_2=I_{M_2}=\langle x_1(x_1+Tx_2)\rangle, ~ \bar{I}_3=\langle x_1 , x_2 \rangle$, and $\bar{I}_4=\bar{I}_3 $. Therefore, according to the Algorithm 2,  $\mathcal{V}(\bar{I}_3)=(0,0)$ is the only non-accessibility singular point, and accessibility index of the system is $r^*=3$.
\end{example}

{\color{black}
 The following example shows how to compute accessibility index and the set of accessibility singular points for a rational system.}
\begin{example}
 Consider the following rational system
 \begin{equation}\label{ex_19}
  x_1^{\langle 1\rangle}  =  \frac{x_2}{u+x_1}, \;\;\;\;
  x_2^{\langle 1\rangle}  =  x_1+x_2.
 \end{equation}
  One can easily verify that this system is generically accessible. 
  If we go through the steps of Algorithm 1, firstly we obtain  $\sqrt[\R]{I_{M_2}}=\la x_2(x_1+x_2)\ra$. Thus $S_2=\{x\in\R^2: x_2=0\;\text{or}\;\;x_1+x_2=0\}$ is the set of points that are non-accessible up to two steps. Next, $I_3=\sqrt[\R]{I_{M_3}}=\la x_1,x_2\ra$, so $S_3=\{x\in\R^2: x_2=0,\;x_1=0\}=\{(0,0)\}$ is the set of points that are non-accessible up to three steps. By computing $I_4=\sqrt[\R]{I_{M_4}}=\la x_1,x_2\ra$ we obtain that $I_4=I_3$. So, according to Theorem \ref{theorem12}, $S_3=S_4=S_\infty=\{(0,0)\}$, and accessibility index is $r^*=3$.
\end{example}
\begin{remark} \label{remark17} 
The problem of  backward accessibility can be tackled using an approach similar to that taken in this paper for the forward accessibility. System (\ref{system}) is called \emph{ backward accessible} to a point $x_0 \in X$ if the set of points in $X$ from which the state can be steered toward $x_0$ in a finite number of steps, contains a nonempty open subset of $X$. In case of reversible systems, including those that arise from sampling of continuous-time systems, it is known that the original system is  backward accessible toward $x_0$ if and only if the time inverse system is forward accessible from $x_0$. As shown in  \cite{Aranda}, reversibility assumption may be relaxed and replaced by submersivity assumption.
\end{remark}
\begin{example}
Consider again the system (\ref{system_after_dyscretization}) for investigation of backward accessibility.  The backward-shift equations of (\ref{system_after_dyscretization}) are
\be  \label{backwardexample} x_1^{\langle -1 \rangle}\!=\frac{1}{\Gamma}[ (bx_1+x_2)T-x_1 ],~~ x_2^{\langle -1 \rangle}\!=\frac{1}{\Gamma}[ u^{\langle -1 \rangle} x_1aT-x_2 ]
\ee
with $\Gamma \coloneqq u^{\langle -1 \rangle } a T^2 +bT-1$. From (\ref{backwardexample}) the inverse (in time) system is obtained as
\be  \label{inverse} z_1^{\langle 1 \rangle}\!=\frac{1}{\Lambda}[ (bz_1+z_2)T-z_1 ],~~  z_2^{\langle 1 \rangle}\!=\frac{1}{\Lambda}[ v z_1aT-z_2 ]
\ee
with $v$ as input and $\Lambda \coloneqq v a T^2 +bT-1$. Based on Remark \ref{remark17}, we obtain the singular points of (forward) accessibility of (\ref{inverse}). The system (\ref{inverse}) is generically forward accessible (in 2 steps), with $I_{M_2}=\langle z_1 (z_1 -T(bz_1 +z_2)) \rangle $. By use of Algorithm 3, we obtain $\bar{I}_2:=I_{M_2}=\langle z_1 (z_1 -T(bz_1 +z_2)), ~ \bar{I}_3=\langle z_1 , z_2 \rangle,~ \bar{I}_4=\bar{I}_3$. In conclusion, the inverse system (\ref{inverse}) is forward accessible from every point in at most 3 steps, except for the point $(0,0)$ which is singular point of forward accessibility for (\ref{inverse}). As a result, the original system (\ref{system_after_dyscretization}) is backward accessible to every point of the state space in at most 3 steps, except for the point $(0,0)$ which is singular point of backward accessibility for (\ref{system_after_dyscretization}). 
\end{example}

\subsection{Analytic case}\label{4.2}

Since the ring of analytic functions on $\R^n$, unlike $\R [x]$, is not  Noetherian, the accessibility index may  fail to be finite, as in Example \ref{ex_11} below. 
 
 \color{black}

\begin{example}\label{ex_11}
 Consider the following system
 \begin{equation}\label{ex_non_rat}
  x^{\langle 1\rangle}=\frac{x}{2}+u{\rm sin}\pi x,\;\;\;x\in\R,\;u\in\R.
 \end{equation}
 In this case, the set of points from which the system (\ref{ex_non_rat}) is not accessible in one step is $S_1=\Z$. Going forward
 and assuming that $x\in S_1$, we obtain that $x^{\langle 1\rangle}=\frac{x}{2}$ and
 $x^{\langle 2\rangle}=\frac{x}{4}+u^{\langle 1\rangle}{\rm sin}\frac{\pi x}{2}.$
 Since $x^{\langle 2\rangle}$ is independent on $u$ and $u^{\langle 1\rangle}$ if and only if $x\in 2\Z$, so  $S_2=2\Z$. Inductively we can see that $S_k=2^{k-1}\Z$, $k\in\N$. So, we obtain infinite sequence $S_1\supset S_2\supset\ldots$ that will never stabilize. 
\end{example}

 Fortunately, the results of the previous subsection can be extended to the analytic case when  the state space  $X$ is assumed to be compact, since the ring of real analytic functions $A_X$ on a compact semianalytic set $X$ is Noetherian \cite{Frish}.
If $I$ is an ideal of $A_X$, then as before $\mathcal{Z}(I)$ denotes the zero-set of $I$ and $\mathcal{I}(\mathfrak{Z})$ denotes the zero-ideal of $\mathfrak{Z}\subset X$, which is now an ideal of $A_X$.

\begin{theorem}\label{Risler}\cite{Risler}
 Let $I$ be an ideal of $A_X$. Then $\mathcal{I}(\mathcal{Z}(I))=\sqrt[\R]{I}$.
\end{theorem}

For a fixed $k>n$ and $\textbf{u}\in U^k$ let $q^i_{k,\textbf{u}}$, $i=1,\ldots,i_k$, denotes determinants of all $n\times n$ submatrices of $M_k(x,\textbf{u})$, where   $i_k$ is the number of minors. For any control sequence $\textbf{u}$ the function $q^i_{k,\textbf{u}}(x)$ is an analytic function parameterized by $\textbf{u}$. Therefore from the definition of $S_k$ and using Lemma \ref{lemma1} and analyticity, $S_k=\mathcal{Z}(I_{M_k})$, where $I_{M_k} \coloneqq \langle q^i_{k,\textbf{u}}:i=1,\ldots,i_k,\textbf{u}\in U^k\rangle$ is the ideal of analytic functions $q^i_{k,\textbf{u}}$.
  Define $I_{k}=\mathcal{I}(S_{k})$. Theorem \ref{Risler} assures that $ I_k=\sqrt[\R]{I_{M_k}}$. Consider the ideals
 $\overline{I}_k :=  \bigcup_{l=n}^k I_{M_l}$
 for any $k\geq n$. By construction, we have the ascending chain $ \overline{I}_{n}\subseteq \overline{I}_{n+1}\subseteq\ldots $ of ideals in $A_X$.

\begin{theorem}\label{theorem2A}
 Assume that the system (\ref{system}) is analytic on a compact semianalytic set $X$ and generically accessible. Then under submersivity assumption (\ref{submersivity})
 \begin{enumerate}
 \item[(i)]  Properties $(i)$ - $(iii)$ of Theorem \ref{theorem2} are valid.
 \item[(ii)] There exists a finite integer $\kappa$ such that $\overline{I}_\kappa=\overline{I}_{\kappa+1}$. Moreover, $r^* \leq \kappa$ and $\mathcal{Z}(\overline{I}_\kappa)=S_{\infty}$.
 \end{enumerate}
\end{theorem}
\begin{proof}
Using the Noetherian property of ${A_X}$ on compact semianalytic set $X$ and Theorem \ref{Risler},  the proofs of the parts $(i)$ and $(ii)$ are similar  to the proofs of   Theorems \ref{theorem2} and \ref{theorem12}, respectively.
\end{proof}
As a consequence of the part $(i)$ of Theorem \ref{theorem2A} we have the ascending chain $I_{n}\subsetneq I_{n+1}\subsetneq\ldots\subsetneq I_{r^*}=I_{r^*+1}=\cdots$ of ideals of $A_X$ with  $I_{k}=\mathcal{I}(S_{k})=\sqrt[\R]{I_{M_k}}$, $k>n$, which means that the accessibility index $r^*$ of the system is the first $k$ for which we have $I_k=I_{k+1}$.
Part $(ii)$  ot Theorem \ref{theorem2A} provides an alternative approach for obtaining an upper bound on $r^*$ along with the set of accessibility singular points in the following way: compute the chain of ideals $\overline{I}_k$ until for some $\kappa$ we reach that $\overline{I}_\kappa=\overline{I}_{\kappa+1}$. Then  $\kappa$ serves as an upper bound on the accessibility index $r^*$, and $S_\infty=\mathcal{Z}(\overline{I}_\kappa)$. 
. 
%
\begin{example}\label{ex_11A}
 Let us come back to system (\ref{ex_non_rat}), but now take $(x,u)\in[0,2]\times [-1,1]$. Using the same reasoning as in Example \ref{ex_11} we obtain $S_1=\{0;2\}$, $S_2=\{0\}=S_3=\ldots$. So, now the sequence $S_1\supset S_2\supset\ldots$ stabilizes.
\end{example}

\subsection{\ Properties of the sets $S_k$} 

The sets $S_k$ are important on their own, as in practice states on the sets $S_k$ for  large enough $k$'s may be considered as singular points of accessibility. Also invariance property of $S_\infty$ allows to reduce the dimension of the system on $S_\infty$, and intuitively in a neighborhood around it.

\begin{theorem}
For a generically accessible system, there exists an integer $k^* \leq n$ such that for  $k\ge k^*$ the set $S_{k}$ is a zero-measure set, and for  $k< k^*$ the set $S_k$ is the entire $X$.
\end{theorem}
\begin{proof}
From Theorem  \ref{theorem1} we know that for some $k^* \leq n$ the matrix $M_{k^*}$ becomes full-rank, which considering the smallest such $k^*$, the matrices $M_k$ for all $k < k^*$ has rank less than $n$ for all $x\in X$ and therefore for every $k< k^*$ the set $S_k$ is the entire $X$. Since $M_{k^*}$ is generically full rank, the set $S_{k^*}$ is a zero-measure set, and by (\ref{descending_seq}) for every $k\ge k^*$ the set $S_{k}$ is a zero-measure set.
\end{proof}
\color{black}

 With our notations, a set $S\subseteq X$ is  a \emph{forward invariant set} for the system (\ref{system}) if for every $x\in S$ it holds that $\Phi_x^i(U^i )\in S$ for all $i\in \N$.

\begin{theorem}
 Assume that the system (\ref{system}) is rational or is analytic on a compact semianalytic set and generically accessible. The set $S_{r^* }$ is a forward invariant set for the system. In addition, any other forward invariant set of measure zero is contained in $S_{r^* }$.
\end{theorem}
\begin{proof}
 The first part was proved in Theorems \ref{theorem2} and \ref{theorem2A} respectively. For the second part, assume that a set 
 $\mathcal{B}\subset X$ is a zero measure set that is also forward invariant for the system (\ref{system}). Let $y\in\mathcal{\mathcal{B}}$. Since $\mathcal{B}$ is forward invariant, then $\mathcal{A}^+(y)\subset\mathcal{B}$, and since $\mathcal{B}$ is a zero measure set, we conclude that ${\rm int}\mathcal{A}^+(y)=\emptyset$. So, the system (\ref{system}) is  non-accessible from any $y\in\mathcal{B}$ and hence $\mathcal{B}\subset S_{r^*}$.
\end{proof}

\begin{remark} \label{remark2}
It was correctly shown in \cite{Kawano_2016} that for any $x_0\in X$ and fixed time step $k$, the set $\mathcal{A}_k^+(x_0)$ is a semialgebraic set.  Unfortunately,  $\mathcal{A}^+(x_0)$ which is the union of all sets $\mathcal{A}_k^+(x_0)$ for $k \in \mathbb{Z}$, is not necessarily so, because an infinite union of semialgebraic sets in not necessarily a semialgebraic set. Therefore  properties of the set $\mathcal{A}^+(x_0)$ may not locally coincide with the properties of its Zariski closure. This point was overlooked in the proof of Lemma 4.1 of \cite{Kawano_2016}, and therefore this Lemma is not correct. For example,
consider the linear discrete-time system $x_1^{\langle 1 \rangle}= u, ~~x_2^{\langle 1 \rangle}=x_2+1$.  For the initial state $x_0=(0,0)$, the set $\mathcal{A}_k^+(x_0)$ for $k\geq 0$ is the line $x_2=k$, and $\mathcal{A}^+(x_0)$ is the union of all lines $x_2=k$, and $\overline{\mathcal{A}^+(x_0)}=\mathbb{R}^2$. Based on Lemma 4.1 of \cite{Kawano_2016} the system must be accessible from $x_0$.  However, the dynamics of $x_2$ is completely independent of $u$, and obviously the system is not accessible anywhere.
\end{remark}

\section{Conclusion}
The paper suggests minimum number of steps needed to decide   (non-)accessibility  for two large subclasses of nonlinear systems - rational systems, and  analytic systems defined on a compact state space. This number  may be far larger than the dimension of the state, unlike in the  case of generic accessibility. As a byproduct, the set of singular  points  $S_\infty$ of accessibility have been  characterized for these  subclasses. From such points the system never becomes accessible. The  set $S_\infty$ can be found as a limit of non-increasing algebraic subsets  $S_k$ of state space. For the subclasses addressed,  the sequence $S_k$  stabilizes in a finite number of steps, thanks to the fact that the   rings of functions defined by such systems are Noetherian. Two  algorithms are given. One of them allows to compute the accessibility  index, i.e. the minimal number of steps to decide accessibility.  Another algorithm is easier to apply but provides only an upper bound  for accessibility index and also computes $S_\infty$.

\section*{Acknowledgment}

The works of Zbigniew Bartosiewicz and Ewa Pawluszewicz have been supported by grants No. S/WI/1/2016 and S/WM/1/2016, respectively, of Bialystok University of Technology, financed by Polish Ministry of Science and Higher Education.



\ifCLASSOPTIONcaptionsoff
  \newpage
\fi



%

\section{Appendix}

%
%
%
%
 Algorithm 1 gives the accessibility index $r^*$ and $I_{r^*}$ (and therefore $S_\infty$).

%

\begin{algorithm}[H] 
\caption{(Finding $r^*$ and $S_\infty$)}
\label{<your label for references later in your document>}
\begin{algorithmic}[1]
\State \textbf{Initialization:} $k \leftarrow n$, $I_{k}\leftarrow \sqrt[\R]{I_{M_n}}$
\State compute $M_{k+1}$ (using (\ref{M_k}))
\State compute $I_{M_{k+1}}$ 
\State compute $\sqrt[\R]{I_{M_{k+1}}}$ and $I_{k+1}\leftarrow \sqrt[\R]{I_{M_{k+1}}}$
\If {$I_k=I_{k+1}$}
    \State stop and return $r^*\leftarrow k$ and $S_\infty \leftarrow \mathcal{V}(I_k)$
\Else
    \State   $k\leftarrow k+1$, $I_k\leftarrow I_{k+1}$ and go to \textit{step 2}
\EndIf
\end{algorithmic}
\end{algorithm}
%
Algorithm 2 gives $\kappa$ (an upper bound for the accessibility index $r^*$), as well as the set $S_\infty$.
\begin{algorithm}[H] 
\caption{(Finding $\kappa$ and $S_\infty$)}
\label{<your label for references later in your document>}
\begin{algorithmic}[1]
\State \textbf{Initialization:} $k \leftarrow n$, $\bar{I}_{k}\leftarrow I_{M_n}$
\State $k \leftarrow n$, $\bar{I}_{k}\leftarrow I_{M_n}$
\State compute $I_{M_{k+1}}$ 
\If {$\bar{I}_k \cup I_{M_{k+1}}=\bar{I}_k$}
    \State stop and return $\kappa \leftarrow k$ and  $S_\infty \leftarrow \mathcal{V}(\bar{I}_k)$
\Else
    \State  $k\leftarrow k+1$, $\bar{I}_k\leftarrow \bar{I}_k \cup I_{M_{k+1}}$ and go to \textit{step 2}
\EndIf
\end{algorithmic}
\end{algorithm}
\color{black}

\end{document}